\documentclass[11pt]{amsart}
\usepackage{graphicx,tikz}

\usetikzlibrary{arrows,shapes,positioning}
\usetikzlibrary{decorations.markings}
\tikzstyle arrowstyle=[scale=1]
\tikzstyle directed=[postaction={decorate,decoration={markings,
    mark=at position .5 with {\arrow[arrowstyle]{stealth}}}}]

\usetikzlibrary{calc,decorations.pathmorphing,patterns,snakes}
\pgfdeclaredecoration{penciline}{initial}{
    \state{initial}[width=+\pgfdecoratedinputsegmentremainingdistance,
    auto corner on length=1mm,]{
        \pgfpathcurveto%
        {
            \pgfqpoint{\pgfdecoratedinputsegmentremainingdistance}
                      {\pgfdecorationsegmentamplitude}
        }
        {
        \pgfmathrand
        \pgfpointadd{\pgfqpoint{\pgfdecoratedinputsegmentremainingdistance}{0pt}}
                    {\pgfqpoint{-\pgfdecorationsegmentaspect
                     \pgfdecoratedinputsegmentremainingdistance}%
                               {\pgfmathresult\pgfdecorationsegmentamplitude}
                    }
        }
        {
        \pgfpointadd{\pgfpointdecoratedinputsegmentlast}{\pgfpoint{1pt}{1pt}}
        }
    }
    \state{final}{}
}

\usepackage{amssymb,amscd}
\usepackage{amsfonts,mathrsfs}
\usepackage{setspace}
\usepackage{enumerate}
\usetikzlibrary{arrows, automata,chains,fit,shapes}
\usepackage{subfig,float}
\usepackage{enumitem}

\newcommand{\ra}{\rightarrow}
\newcommand{\e}{\varepsilon}

\newcommand{\Ra}{\Rightarrow}
\newcommand{\cyc}{\operatorname{cyc}}

\newcommand\A{\mathbf A}
\newcommand\B{\mathbf B}
\newcommand\C{\mathbf C}

\newcommand\M{\mathbf M}

\newcommand{\X}{\mathbf{X}}
\newcommand{\Y}{\mathbf{Y}}
\newcommand{\Z}{\mathbf{Z}}
\newcommand{\ol}{\overline}

\newtheorem{thm}{Theorem}
\newtheorem{prop}[thm]{Proposition}
\newtheorem{lem}[thm]{Lemma}
\newtheorem{cor}[thm]{Corollary}
\newtheorem{defn}[thm]{Definition}

\newcommand{\s}{\mathbf{S}}

\newcommand{\Aho}{MR0258547}

\newcommand{\HU}{MR645539}
\newcommand{\Brandstadt}{MR630064}
\newcommand{\Oshiba}{MR0478788}
\newcommand{\Maslov}{MR0334597}

\begin{document}

\title{
Permutations of  context-free and indexed languages}
\author{Tara Brough}
\address{School of Mathematics and Statistics,
University of St Andrews,
North Haugh,
St Andrews
KY16 9SS, Scotland}
\email{t.brough@st-andrews.ac.uk}

\author{Laura Ciobanu}
\address{Mathematics Department, University of Neuch\^atel,
Rue Emile - Argand 11, CH-2000 Neuch\^atel, Switzerland}
\email{laura.ciobanu@unine.ch}

\author{Murray Elder}
\address{School of Mathematical and Physical Sciences,
The University of Newcastle,
Callaghan NSW 2308, Australia}
\email{murray.elder@newcastle.edu.au}

\keywords{indexed grammar; context-free grammar;  cyclic closure}  
\subjclass[2010]{20F65; 68Q45}
\date{\today}
\thanks{Research supported by  Australian Research Council grant  FT110100178, Swiss National Science 
Foundation Professorship FN PP00P2-144681/1 and London Mathematical Society Scheme 4 grant 41348}

\begin{abstract}
We consider the cyclic closure of a language, and its
 generalisation to the operators $C^k$ introduced by Brandst\"adt.
 We prove that the cyclic closure of an indexed language is indexed, and 
 that if $L$ is a context-free language then $C^k(L)$ is indexed.
\end{abstract}

\maketitle

\section{Introduction}

The {\em cyclic closure} of a language $L$ is the language  
\[\cyc(L)=\{w_2w_1\mid w_1w_2\in L\}.\] It is easy to show that the classes of regular, context-sensitive and recursively enumerable languages are closed under this operation.
Maslov and independently Oshiba \cite{\Maslov, \Oshiba}
proved that the class of  context-free languages is also closed under this operation. In this article we show that in addition,  the class of   indexed languages is closed under taking cyclic closure.

Brandst\"adt \cite{\Brandstadt} generalised the notion of  cyclic closure to a family of operators on languages $C^k$ for $k\in\mathbb N$ by defining
\[C^k(L)=\{w_{\sigma(1)}\ldots w_{\sigma(k)}\mid w_1\ldots w_k\in L,\sigma\in S_k\},\]
where $S_k$ is the set of permutations on $k$ letters.
So $C^2(L)$ is exactly the cyclic closure.
He proved that if $L$ is context-free and $k\geq 3$  the language  $C^k(L)$ 
is not context-free in general, while it is always context-sensitive.
Here we sharpen this result by showing that $C^k(L)$  is indexed.

A natural generalisation of context-free and indexed languages was given by Damm and co-authors in \cite{MR666544,MR864744}, where they defined the OI- and  IO-hierarchies of languages built out of  automata or grammars that extended the pushdown automata and indexed grammars, respectively. They define level-$n$ grammars inductively, allowing the flags at level $n$ to carry up to $n$ levels of parameters in the form of flags. Thus level-$0$ grammars generate context-free languages, and level-$1$ grammars produce indexed languages.
We conjecture that the class of level-$n$ languages is closed under cyclic closure, and also that 
if $L$ is an level-$n$ language then $C^k(L)$ is an level-$(n+1)$ language. 
This paper is the first step in proving this conjecture and completing the picture of cyclic closure and permutation operators for the OI- and IO-hierarchies.

\section{Preliminaries}\label{sec:prelim}

\subsection{Permutation operators}\label{permclose}

Brandst\"adt~\cite{\Brandstadt} defined  the language $C^k(L)$ to be the set of all words obtained from $L$ by permutating $k$ subwords according to some permutation.
In this article we specialise the definition to individual permutations as follows.

\begin{defn}[Permutation operator]
Let $k$ be a positive integer, $\sigma\in S_k$ a permutation on $\{1,\ldots,k\}$, and $L\subseteq \Sigma^*$ a language over a finite alphabet.
The  language $\sigma(L)$ is defined by $$\sigma(L)=\{w_{\sigma(1)}\cdots w_{\sigma(k)} \mid w_1\cdots w_k\in L\}.$$ 
\end{defn}

If $w=w_{\sigma(1)}\dots w_{\sigma(k)}\in\sigma(L)$ where $w_1\dots w_k\in L$, 
 some subwords $w_i$ could be empty.  We can write $w_1\ldots w_k$ as $u_1\ldots u_\ell$ 
for some $\ell\leq k$, where each $u_i$ is equal to some non-empty $w_j$.  Then there is a permutation 
$\tau\in S_\ell$ (called a {\em subpattern} of $\sigma$) such that $w= u_{\tau(1)} \ldots u_{\tau(\ell)}$.

For $\tau\in S_\ell$, define 
\[L_\tau = \{ w_{\tau(1)} \ldots w_{\tau(\ell)} \mid w_1\ldots w_\ell\in L, w_i\neq \e \; (1\leq i\leq \ell)\}. \]
Then $\sigma(L) = \bigcup_\tau L_\tau$ with $\tau$ ranging over all subpatterns of $\sigma$.
Thus if ${\mathcal C}_1$ and ${\mathcal C}_2$ are two language classes, with ${\mathcal C}_2$ closed under finite union, 
and we wish to show that $\sigma(L)\in {\mathcal C}_2$ for all $L\in {\mathcal C}_1$ and $\sigma\in S_k$, 
it suffices to show that $L_\tau\in {\mathcal C}_2$ for all $L\in {\mathcal C}_1$ for all $\tau\in \bigcup_{1\le \ell \le k} S_\ell$.

For any language $L$ and $k\in \mathbb{N}$, we have 
$$C^k(L) = \bigcup_{\sigma\in S_k} \sigma(L) =  \bigcup_{1\leq \ell\leq k}\bigcup_{\sigma\in S_\ell} L_\sigma.$$

Note that the language $L_\tau$ does not in general contain $L$ as a sublanguage, whereas $\sigma(L)$ contains $L$ since we may take all but one subword to be empty.

\subsection{Indexed languages}

We define an indexed language to be one that is generated by the following type of grammar:
\begin{defn}
[Indexed grammar; Aho \cite{\Aho}]
An {\em indexed grammar}  is a 5-tuple $(\mathcal N, \mathcal T, \mathcal I, \mathcal P, \mathbf{S})$ such that
\begin{enumerate}
\item $\mathcal N, \mathcal T, \mathcal I$ are three mutually disjoint sets of symbols, called {\em nonterminals, terminals} and {\em indices} (or {\em flags}) respectively.
\item $\mathbf S\in\mathcal N$ is the {\em start symbol}.
\item $\mathcal P$ is a finite set of {\em productions}, each having the form of one of the following:
\begin{enumerate}
\item  $\mathbf{A} \ra \mathbf{B}^f$.
\item $\mathbf{A}^f \ra v$.
\item  $\mathbf{A} \ra u$.
\end{enumerate}
where $\mathbf{A}, \mathbf{B} \in\mathcal N$, $f\in \mathcal I$ and $u,v\in(\mathcal N\cup\mathcal T)^*$.
\end{enumerate}
\end{defn}
The language defined by an indexed grammar is the set of all strings of terminals that can be obtained by successively 
applying production rules starting from a rule which has the start symbol $\mathbf S$ on the left. Production rules operate as follows. 
Let $\A\in \mathcal N, \gamma\in \mathcal I^*$ and suppose $\mathbf{A}^\gamma$ appears in some string.
\begin{enumerate}
\item applying $\mathbf{A} \ra \mathbf{B}^f$ replaces  $\mathbf{A}^\gamma$ by $\mathbf{B}^{f \gamma}$
\item if $\gamma=f \delta$ with $f\in \mathcal I$, applying $\mathbf{A}^f \ra \mathbf Ba \mathbf C$ replaces $\mathbf{A}^\gamma$ by 
$\mathbf B^\delta a \mathbf C^\delta$
\item applying  $\mathbf{A} \ra  \mathbf Ba \mathbf C$ replaces $\mathbf{A}^\gamma$ by  $\mathbf B^\gamma a \mathbf C^\gamma$.\end{enumerate}
We call the operation of  successively applying productions starting from one which has the start symbol $\mathbf S$ on the left and terminating at a string $u\in\mathcal T^*$  a {\em derivation} of $u$. We use the notation $\Rightarrow$ to denote a sequence of productions within a derivation, and call such a sequence a {\em subderivation}.

\begin{defn}[Normal form]
An indexed grammar $(\mathcal N, \mathcal T, \mathcal I, \mathcal P, \mathbf{S})$ is in {\em normal form} if
\begin{enumerate}
\item the start symbol only appears on the left side of a production,
\item  productions are of the following type:
\begin{enumerate}
\item  $\mathbf{A} \ra \mathbf{B}^f$
\item $\mathbf{A}^f \ra \mathbf B$
\item \label{third} $\mathbf{A} \ra \mathbf{BC}$
\item  \label{forth} $\mathbf{A} \ra a$
\end{enumerate}
where $\mathbf{A}, \mathbf{B}, \mathbf C \in\mathcal N$, $f\in \mathcal I$ and $a\in\mathcal T$.
\item $\mathcal I$ contains a special `end-of-flag' symbol $\$ $, and the start symbol begins with the 
string $\$ $ in its flag.  The symbol $ \$ $ is otherwise never used in any $\A \ra \B^f$ production.
\end{enumerate}
\end{defn}

An indexed grammar can be put into normal form as follows. 
Introduce a new nonterminal $\mathbf S_0$, a production $\mathbf S_0\ra\mathbf S$, and declare $\mathbf S_0$ to be the new start symbol.
This ensures condition (1).
For each production $\A^f \ra v$ with $v\not\in\mathcal N$, introduce a new nonterminal $\B$, add productions  $\A^f \ra \B, \B \ra v$, and remove $\A^f \ra v$.
By the same arguments used for Chomsky normal form, each production  $\mathbf{A} \ra u$ without flags can be replaced by a set of 
productions of type \ref{third} and \ref{forth} above.
Instead of beginning derivations with the start symbol $\s_0$ (with empty flag), begin with $\s_0^\$ $.  
(Introducing the symbol $\$ $ in this way into an existing indexed grammar is pointless, but harmless.
For constructing new grammars, however, it is often very useful to have a way of telling when a flag is `empty'.)

In an  indexed grammar in normal form, every nonterminal in a derivation has a flag of the form $\gamma\$ $ where $\gamma\in\mathcal I^*$. 
The symbol $\$$ is removed only by productions of type  \ref{forth}.

\subsection{Tree-shapes and parse tree skeletons}
A {\em parse tree} in an indexed or context-free grammar is a standard way to represent a derivation in the grammar, see for example \cite{\HU}. 
In this paper, all trees will be rooted, and will be regarded as being drawn in the plane, with the root at the top, and
with a fixed orientation.  
For a tree $T$, the {\em shape} of $T$ is the tree $\hat{T}$ obtained from $T$ as follows:
 \begin{enumerate}
 \item add an edge to the root vertex, so that the root has degree 1
 \item delete all vertices of degree $2$
  \end{enumerate}
 A {\em tree-shape} is hence a tree with no vertices of degree $2$, and root degree 1.
For example, the two possible  tree-shapes with $3$ leaves 
are shown in Figure~\ref{fig:shapes}.

\begin{figure}[h!]
 {\begin{tikzpicture}[scale=.5]
\draw[decorate] (1,1) -- (3,3)--(3,4.4);
\draw[decorate]  (2,2) -- (3,1);
\draw[decorate]  (3,3) -- (4,2);
\end{tikzpicture}}
 \hspace{5mm}
 {\begin{tikzpicture}[scale=.5]
\draw[decorate]  (2,2) -- (3,3)--(3,4.4);
\draw[decorate]  (4,2) -- (3,1);
\draw[decorate]  (3,3) -- (5,1);
\end{tikzpicture}}
\caption{Possible tree-shapes with $3$  leaves.}
\label{fig:shapes} 
\end{figure}
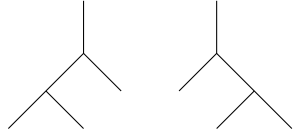

Let $P$ be a parse tree in some grammar and $F$ a subtree of $P$ (with all labels preserved). Let $B$ be the $1$-neighbourhood\footnote{the subtree of $P$ consisting of all edges with at least one end vertex in $F$}
of $F$, and suppose $F$ has  tree-shape $T$.
We call $B$ a {\em $T$-skeleton}
of $P$, and $F$ the {\em frame} of $B$.

If $\Gamma$ is a grammar, then we call a tree $B$ labelled by symbols from $\Gamma$ a {\em $T$-skeleton in $\Gamma$}
if $B$ is a $T$-skeleton of some parse tree in $\Gamma$.
For example, if $T$ is the first of the tree-shapes in Figure~\ref{fig:shapes}, then Figure~\ref{fig:skeleton} is a $T$-skeleton 
in some indexed grammar, with the frame $F$ marked in bold.

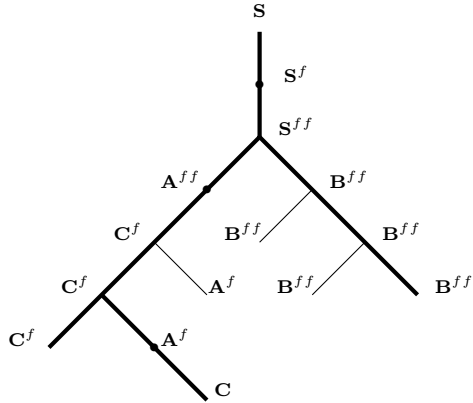
\begin{figure}[h!]   \tiny
\begin{tikzpicture}[scale=.7]
\draw[ultra thick, decorate] (0,1) -- (4,5) -- (4,7);
\draw[ultra thick, decorate] (4,5) -- (7,2);
\draw[ultra thick, decorate] (1,2) -- (3,0);
\draw[decorate] (2,3) -- (3,2);
\draw[decorate] (4,3) -- (5,4);
\draw[decorate] (5,2) -- (6,3);

\draw (2,1) node {$\bullet$};
\draw (3,4) node {$\bullet$};
\draw (4,6) node {$\bullet$};

\draw (3.3,2.2) node {$\A^{f}$};
\draw (2.4,1.2) node {$\A^{f}$};
\draw (3.3,.2) node {$\C$};

\draw (5.7,4.2) node {$\B^{ff}$};
\draw (6.7,3.2) node {$\B^{ff}$};
\draw (7.7,2.2) node {$\B^{ff}$};
\draw (3.7,3.2) node {$\B^{ff}$};
\draw (4.7,2.2) node {$\B^{ff}$};

\draw (4,7.4) node {$\s$};
\draw (4.7,6.2) node {$\s^{f}$};
\draw (4.7,5.2) node {$\s^{ff}$};
\draw (2.5,4.2) node {$\A^{ff}$};
\draw (1.5,3.2) node {$\C^f$};
\draw (.5,2.2) node {$\C^f$};
\draw (-.5,1.2) node {$\C^f$};
\end{tikzpicture}
\caption{A $T$-skeleton  in an indexed grammar.
\label{fig:skeleton}}
\end{figure}

If $F$ is a path, then the shape of $F$ is an edge, and so we call a skeleton with frame $F$ an {\em edge-skeleton}.

\section{Main results}\label{sec:main}

Brandst\"adt proved that the classes of regular, context-sensitive and recursively enumerable languages are closed 
under the operation $C^k$ for all $k$ \cite{\Brandstadt}.
We start by  reproving Theorem 1 of \cite{\Brandstadt}, modified for $\sigma(L)$.
\begin{lem}
If $L$ is regular then $\sigma(L)$ is regular for any fixed permutation $\sigma$.
\end{lem}
\begin{proof}
Assume $L$ is the language of a finite state automaton $M$ with start state $q_{\mathrm{start}}$ and single accept state $q_{\mathrm{accept}}$, and  $\sigma\in S_k$.
For each $(k-1)$-tuple of states $\mathbf q=(q_{j_1},\dots, q_{j_{k-1}})$, define  $k$ automata $M^{\mathbf q}_1,\dots, M^{\mathbf q}_k$ as follows. 
Let $M^{\mathbf q}_1$ be a copy of $M$ with start state $q_{\mathrm{start}}$ and accept state $q_{j_1}$. For $1\leq s< k-1$ let $M^{\mathbf q}_{s+1} $ be a copy of $M$ with start state $q_{j_s}$ and accept state $q_{j_{s+1}}$. 
Let $M^{\mathbf q}_k$ be a copy of $M$ with start state $q_{j_{k-1}}$ and accept state $q_{\mathrm{accept}}$.

Define $\overline{M^{\mathbf q}_i}$ to be the language accepted by the automaton ${M}^{\mathbf q}_i$, and let $L_{\mathbf q}$ be the concatenation
$$\overline {M^{\mathbf q}_1} \  \overline {M^{\mathbf q}_2}\dots \overline {M^{\mathbf q}_k}.$$

 Then $L_{\mathbf q}$ accepts precisely the words in  $L$ that label a path in $M$ from $q_{\mathrm{start}}$ to $q_{\mathrm{accept}}$ that passes the intermediate states from $\mathbf q$.
 It follows that $L=\bigcup_{\mathbf q} L_{\mathbf q}$. 
 Now define 
$L_{\mathbf q}^\sigma$ to be the concatenation $$\overline {M^{\mathbf q}_{\sigma(1)}} \ \overline {M^{\mathbf q}_{\sigma(2)}}\dots \overline {M^{\mathbf q}_{\sigma(k)}}.$$  Then $w\in L_{\mathbf q}^\sigma$ if and only if $w=w_{i_1}\dots w_{i_k}$ and $$w_{\sigma^{-1}(i_1)}\dots w_{\sigma^{-1}(i_k)}\in L_{\mathbf q}.$$  It follows that $\sigma(L)=\bigcup_{\mathbf q} L_{\mathbf q}^\sigma$.
\end{proof}

Maslov and independently Oshiba \cite{\Maslov, \Oshiba}
proved that the cyclic closure of a context-free language is context-free.
A sketch of a proof of this fact is given in the solution to Exercise 6.4 (c) in \cite{MR645539},
and we generalise the approach taken there to show that the class of indexed languages is also closed 
under the cyclic closure operation.

\begin{thm}\label{thm:firstmain}
If $L$ is indexed, then $\cyc(L)$ is indexed. 
\end{thm}

\begin{proof}
The idea of the proof is to take the parse-tree of a derivation of $w_1w_2\in L$ in $\Gamma$ and ``turn it upside down",
using the leaf corresponding to the first letter of the word $w_2$ as the new start vertex. 

Let $\Gamma = (\mathcal N, \mathcal X, \mathcal I, \mathcal P, \s)$ be an indexed grammar for $L$ in normal form.
If $w = a_1\ldots a_n\in L$ with $a_i\in \mathcal X$ and we wish to generate the cyclic permutation $a_k\ldots a_n a_1\ldots a_{k-1}$ of $w$, 
take some parse tree for $w$ in $\Gamma$ and draw the unique path $F$ from the start symbol $\s^\$$ to $a_k$.
Consider the edge-skeleton of this parse tree with frame $F$.

\begin{figure}[h!]   \tiny

\begin{tikzpicture}[scale=.7]
\draw[ultra thick,decorate]  (2,7) -- (4,5) -- (1,2)--(3,0);

\draw[decorate] (2,7) -- (1,6);
\draw[decorate] (4,5) -- (5,4);
\draw[decorate] (3,4) -- (4,3);
\draw[decorate] (1,2) -- (0,1);

\draw (3,6) node {$\bullet$};
\draw (2,3) node {$\bullet$};
\draw (2,1) node {$\bullet$};

\draw (2.1,7.4) node {$\s^\$$};
\draw (0.7,6.2) node {$\A_1^\$$};
\draw (3.3,6.2) node {$\B_1^\$$};
\draw (4.4,5.2) node {$\B_2^{f\$}$};
\draw (5.3,4.2) node {$\A_4^{f\$}$};
\draw (2.7,4.2) node {$\B_3^{f\$}$};
\draw (4.3,3.2) node {$\A_3^{f\$}$};
\draw (1.7,3.2) node {$\B_4^{f\$}$};
\draw (0.7,2.2) node {$\B_5^{gf\$}$};
\draw (-.3,1.2) node {$\A_2^{gf\$}$};
\draw (2.5,1.2) node {$\B_6^{gf\$}$};
\draw (3,-.2) node {$a_k$};

\draw (3,-2.37) node{$\;$};

\draw (7.5,4) node{\large$\leadsto$};
\end{tikzpicture}
\hspace{3mm}
\begin{tikzpicture}[scale=.9]
\draw[ultra thick, decorate] (1,0) -- (1,7);

\draw[decorate] (1,7) -- (0.2,6.2);
\draw[decorate] (1,5) -- (1.8,4.2);
\draw[decorate] (1,4) -- (1.8,3.2);
\draw[decorate] (1,2) -- (0.2,1.2);

\draw (1,6) node {$\bullet$};
\draw (1,3) node {$\bullet$};
\draw (1,1) node {$\bullet$};

\draw (1.1,7.3) node {$\s^\$$};
\draw (0,6.4) node {$\A_1^\$$};
\draw (1.3,6.2) node {$\B_1^\$$};
\draw (0.65,5.2) node {$\B_2^{f\$}$};
\draw (2,4.2) node {$\A_4^{f\$}$};
\draw (0.65,4.2) node {$\B_3^{f\$}$};
\draw (2 ,3.2) node {$\A_3^{f\$}$};
\draw (0.65,3.2) node {$\B_4^{f\$}$};
\draw (1.4,2.2) node {$\B_5^{gf\$}$};
\draw (-.15,1.2) node {$\A_2^{gf\$}$};
\draw (1.4,1.2) node {$\B_6^{gf\$}$};
\draw (1,-.2) node {$a_k$};

\end{tikzpicture}

\caption{Edge-skeleton  in an indexed grammar. 
The right-hand version is the same tree as the left-hand one, but `straightened' along the path from $\s^\$ $ to $a_k$.
\label{fig:edgeskeleton}}
\end{figure}
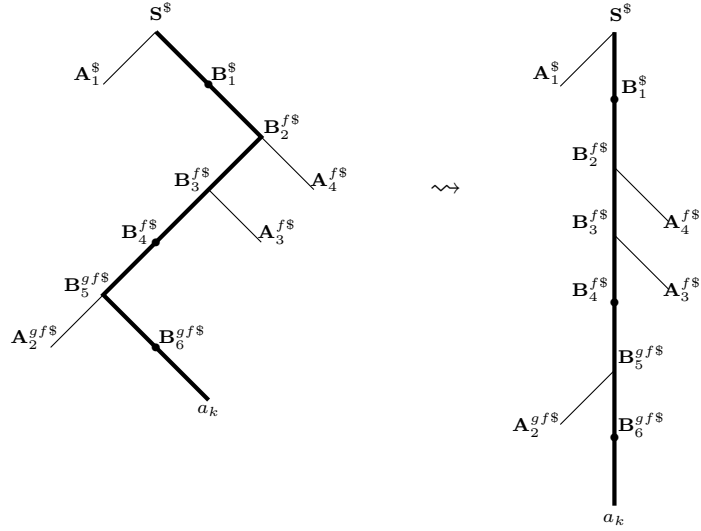

In the example given in Figure~\ref{fig:edgeskeleton},
the desired word $a_k\ldots a_n a_1\ldots a_{k-1}$ can be derived from the 
string $a_k \A_3^{f\$} \A_4^{f\$} \A_1^\$ \A_2^{gf\$}$, using productions in $\mathcal P$.  

Therefore we wish to enlarge the grammar to generate all strings $$a_k \A_{k+1}^{w_{k+1}}\dots \A_n^{w_n} \A_1^{w_1} \ldots \A_{k-1}^{w_{k-1}},$$
where $\A_1^{w_1},\ldots,\A_{k-1}^{w_{k-1}}$ are the labels of the vertices lying immediately to the left of $F$ (in top to bottom order), 
and $\A_{k+1}^{w_{k+1}},\ldots,\A_n^{w_n}$ are the labels of the vertices lying immediately to the right of $F$ (in bottom to top order).
We do this by introducing new `hatted' nonterminals, with which we label all the vertices along the path $F$,
and new productions which are the reverse of the old productions `with hats on'.  By first nondeterministically guessing
the flag on the nonterminal immediately preceding $a_k$, we are able to essentially generate the edge-skeleton in reverse.

The grammar for $\cyc(L)$ is given by $(\mathcal N', \mathcal T', \mathcal I', \mathcal P', \mathbf{S}_0)$, 
where $\mathcal T'=\mathcal T$, $\mathcal I'=\mathcal I$, $\s_0$ is the new start symbol, and $\mathcal N'$ and $\mathcal P'$ are as follows.
Let $\hat{\mathcal N}$ be the set of symbols obtained from $\mathcal N$ by placing a hat on them. 
Then $\mathcal N' = \mathcal N\cup \hat{\mathcal N}\cup\{\mathbf{S}_0, \tilde{\s}\}$ is the new set of nonterminals.

The productions $\mathcal P'$ are obtained as follows:
\begin{itemize}
\item keep all the old productions from $\mathcal P$.
\item add productions $\s_0 \ra \s$, $\s_0 \ra \tilde{\s}$, $\hat{\s}^{\$}\ra \e$ 
\item for each $f\in \mathcal I$, add a production $\tilde{\s} \ra \tilde{\s}^f$
\item for each production $\A \rightarrow a$ in $\mathcal P$, add a production $\tilde{\s} \rightarrow a\hat{\A}$

\item for each production $\mathbf A\ra \mathbf B^f$ in $\mathcal P$, add a production
$\hat{\mathbf B}^f \ra \hat{\mathbf A}$

\item for each production $\A^f \ra \B$ in $\mathcal P$, add a production 
$\hat{\B} \ra \hat{\A}^f$

\item for each production $\A \ra \B\C$ in $\mathcal P$, add productions
$\hat{\B} \ra \C \hat{\A}$ and $\hat{\C} \ra \hat{\A} \B$.
\end{itemize}

Note that the new grammar is no longer in normal form.
Also note that the only way to remove the hat symbol is to apply the production $\hat{\s}^{\$}\ra \e$. 

We will show by induction that in this new grammar, \begin{equation}\A^w \Ra \A^{w_1}_1\ldots  \A^{w_i}_i \ldots \A^{w_n}_n \end{equation} if and only if 
 \begin{equation}\hat{\A}^{w_i}_i \Ra \A^{w_{i+1}}_{i+1} \ldots \A^{w_n}_n \hat{\A}^w \A^{w_1}_1 \ldots \A^{w_{i-1}}_{i-1}  \end{equation} for all $1\leq i\leq n$.

To see why this will suffice, suppose first that \[ \s^{\$} \Ra \A^{w_1}_1\ldots \A^{w_{i-1}}  \A^{w_i}_i \A^{w_{i+1}}_{i+1} \ldots \A^{w_n}_n  \ra \A^{w_1}_1\ldots \A^{w_{i-1}} a \A^{w_{i+1}}_{i+1} \ldots \A^{w_n}_n \] in the original grammar. 
So $\A_i\ra a $ is in $\mathcal P$.
Then in the new grammar
\begin{align*}
 \s_0^{\$} \Ra \tilde{\s}^{w_i}  \ra a \hat{\A}^{w_i}_i & \Ra a \A^{w_{i+1}}_{i+1} \ldots \A^{w_n}_n \hat{\s}^{\$} \A^{w_1}_1 \ldots \A^{w_{i-1}}_{i-1} \\
&  \ra a \A^{w_{i+1}}_{i+1} \ldots \A^{w_n}_n \A^{w_1}_1 \ldots \A^{w_{i-1}}_{i-1},
\end{align*}
hence every cyclic permutation of a word in $L$ is in the new language.

Conversely, suppose $\s_0^{\$} \Ra a \B^{v_1}_1 \ldots \B^{v_n}_n$ and that this subderivation does not start with  $\s_0^{\$} \ra \s^{\$}$.
Then the subderivation begins with $\s_0^{\$} \ra \tilde{\s}^{\$} \Ra \tilde{\s}^u \ra a \hat{\A}^u$ for some $u\in {\mathcal I}^*$, $\A\in \mathcal N$.
Once a `hatted' symbol has been introduced, the only way to get rid of the hat is via the production $\hat{\s}^{\$} \ra \e$.
Hence we must have $\hat{\A}^u \Ra \B^{v_1}_1 \ldots \B^{v_j}_j \hat{\s}^{\$}  \B^{v_{j+1}}_{j+1} \ldots \B^{v_n}_n$ for some ${0\leq j\leq n}$
(with the subword before or after $\hat{\s}$ being empty if $j=0$ or $j=n$ respectively).

But then 
\begin{align*}
\s^{\$} & \Ra \B^{v_{j+1}}_{j+1} \ldots \B^{v_n}_n \A^u \B^{v_1}_1 \ldots \B^{v_j}_j \\
& \ra \B^{v_{j+1}}_{j+1} \ldots \B^{v_n}_n a \B^{v_1}_1 \ldots \B^{v_j}_j 
\end{align*}
and so if a word is produced by the new grammar, some permutation of that word is in $L$.

We finish by giving the inductive proof of the equivalence of (1) and (2).
For the base case, we have $\A^w \Ra \B^u \C^v$ if and only if at some point in the parse tree, 
there is a production $\X^t \ra \Y^t \Z^t$, 
with $\A^w \Ra \X^t$, $\Y^t \Ra \B^u$ and $\Z^t \Ra \C^v$.  The productions in these last three subderivations are all of the form
$\mathbf{D} \ra \mathbf{E}^f$ or $\mathbf{D}^f \ra \mathbf{E}$, so they are equivalent to 
$\hat{\X}^t \Ra \hat{\A}^w$, $\hat{\B}^u \Ra \hat{\Y}^t$ and $\hat{\C}^v \Ra \hat{\Z}^t$.  
Also $\X \ra \Y \Z$ if and only if $\hat{\Y} \ra \Z \hat{\X}$ and $\hat{\Z} \ra \hat{\X} \Y$.
Putting these together, we have $\A^w \Ra \B^u \C^v$ if and only if
\[ \hat{\B}^u \Ra \hat{\Y}^t \ra \Z^t \hat{\X}^t \Ra \C^v \hat{\A}^w \]
and
\[ \hat{\C}^v \Ra \hat{\Z}^t \ra \hat{\X}^t \Y^t \Ra \hat{\A}^w \B^u, \]
as required.

Now for $k>2$, suppose our statement is true for $n<k$.
Then $\A^w \Ra \A^{w_1}_1 \A^{w_2}_2 \ldots \A^{w_k}_k$ if and only if for each $1\leq i\leq k$
there are $\X_i, \Y_i, \Z_i\in {\mathcal N}$ and $t\in {\mathcal I}^*$ such that $\X_i\ra \Y_i \Z_i$ and 
for some $1\leq j\leq k$ either
\[ \A^w \Ra \A^{w_1}_1 \ldots \A^{w_{i-1}}_{i-1} \X_i^t \A^{w_j}_j \ldots \A^{w_k}_k, \]
with $\Y_i^t \Ra A_i^{w_i}$ and $\Z_i^t \Ra \A_{i+1}^{w_{i+1}} \ldots \A_{j-1}^{w_{j-1}}$, or
\[ \A^w \Ra \A^{w_1}_1 \ldots A^{w_j}_j \X_i^t \A^{w_{i+1}}_{i+1} \ldots \A^{w_k}_k, \]
with $\Y_i^t \Ra \A_{j+1}^{w_{j+1}} \ldots \A_{i-1}^{w_{i-1}}$ and $\Z_i^t \Ra \A_i^{w_i}$.

We will consider only the second of these, as it is the slightly more complicated one and the first is 
very similar.  
The right hand side of the displayed subderivation has fewer than $k$ terms, so by our assumption, this
subderivation is valid if and only if 
\[ \hat{\X}_i^t \Ra \A^{w_{j+1}}_{j+1} \ldots \A^{w_k}_k \hat{\A}^w \A^{w_1}_1 \ldots \A^{w_{i-1}}_{i-1}. \]
But this, together with $\Y_i^t \Ra \A_{j+1}^{w_{j+1}} \ldots \A_{i-1}^{w_{i-1}}$ and $\Z_i^t \Ra \A_i^{w_i}$,
is equivalent to
\[ \hat{\A}^{w_i}_i \Ra \hat{\Y}^t \ra \Z^t \hat{\X}^t \Ra \A^{w_{i+1}}_{i+1} \ldots \A^{w_k}_k \hat{\A}^w \A^{w_1}_1 \ldots \A^{w_{i-1}}_{i-1}. \]
\end{proof}

Next, we now show that when $L$ is context-free,  $L_\sigma$ is indexed. 
Since Brandst\"adt proved that the class of context-free languages is not closed under $C^k$ for all $k\geq 3$, and 
$$C^k(L) = \bigcup_{\sigma\in S_k} \sigma(L) =  \bigcup_{1\leq \ell\leq k}\bigcup_{\sigma\in S_\ell} L_\sigma,$$
we have that for all $k\geq 3$ there exist permutations $\sigma \in S_k$
such that $L_\sigma$ are not context-free for some context-free language $L$. 

\begin{prop}\label{prop:tau}
Let $\tau\in S_\ell$ be a permutation.  If $L$ is context-free, then $L_\tau$ is indexed.
\end{prop}
\begin{proof}
The proof is based partly on a similar idea to the proof of Theorem~\ref{thm:firstmain}, except that since we 
are splitting our words up into $\ell$ subwords rather than only two, we need to consider 
skeletons with frames having more complicated shapes than just a single edge.

For $w = w_1\ldots w_\ell\in L$ with all $w_i$ non-empty, let $x_i$ be the first symbol of $w_i$,
and consider a parse tree skeleton with frame consisting of the unique paths from the start symbol
$\s^\$$ to each $x_i$ for $2\leq i\leq \ell$ (these paths will generally overlap).  An example is shown in Figure~\ref{fig:Ltau_skeleton}.

\begin{figure}[h!]\tiny
 \begin{tikzpicture}[scale=.5]             
\draw[decorate] (0,0) -- (6,6)--(6,8);
\draw[decorate] (3,3) -- (7,-1);
\draw[decorate] (6,6) -- (8,4);

\draw[decorate] (6,7) -- (6.7,6.3);
\draw[decorate] (7,5) -- (6.3,4.3);
\draw[decorate] (5,5) -- (4,4.8);
\draw[decorate] (4,4) -- (4.7,3.3);
\draw[decorate] (2,2) -- (1,1.8);
\draw[decorate] (1,1) -- (0,.8);
\draw[decorate] (4,2)--(5,1.8);
\draw[decorate] (5,1) -- (4.3,.3);
\draw[decorate] (6,0) -- (7,-.2);

\draw (6.2,8.4) node {$\s^\$$};
\draw (0,-.2) node {$x_2$};
\draw (7.3,-1.2) node {$x_3$};
\draw (8.3,3.8) node {$x_4$};

\end{tikzpicture}
\caption{Skeleton for $L_\tau$.
\label{fig:Ltau_skeleton}}
\end{figure}
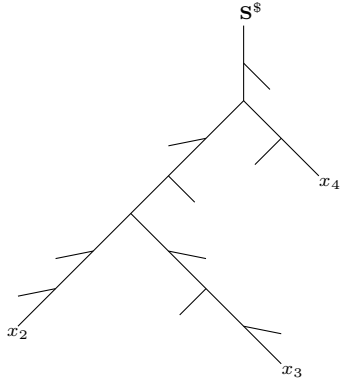

A skeleton $B$ defines a sublanguage of $L$ -- namely all those words which can be generated by completing $B$ into a full 
parse tree -- as well as a fixed $\ell$-partition of words in this language.
Let $L(B)$ be the set of all `partitioned words' $w_1|\ldots|w_\ell$ (where $w_1 \ldots w_n\in L$) generated from the skeleton $B$.
Let $L_\tau(B)$ be the set of all words $w_{\tau(1)}\ldots w_{\tau(\ell)}$ such that $w_1|\ldots|w_\ell\in L(B)$.
Then $L_\tau$ is the union of all the (infinitely many) languages $L_\tau(B)$.

Our grammar for $L_\tau$ will be based on constructing all possible
$(\ell-1)$-leaved skeletons in $L$.  There are finitely many possible shapes for the frames of these skeletons,
and we will construct one grammar for each tree-shape with $\ell-1$ leaves.

We shall abuse standard terminology somewhat, by referring to an edge as a {\em leaf} if one of its vertices 
is a leaf, and talking about {\em parent, child} and {\em sibling} edges, the usual convention being to use these terms
for vertices.

Let $T$ be a tree-shape with $\ell - 1$ leaves.  We construct a grammar $\Gamma_T$ that will produce 
$\bigcup_{B} L_\tau(B)$, where $B$ ranges over all $T$-skeletons for words in $L$.
First place an ordering $e_1,e_2,\ldots,e_N$ on the set $E(T)$ of edges of $T$ (where $N=2\ell-3$), 
such that $e_1$ is the edge with one vertex at the root, 
and for some $1\leq m\leq N$ the edges $e_1,\ldots,e_m$ are all non-leaves, while the remaining edges are leaves.
Choose the ordering such that a parent edge always comes before its children, and a left sibling always comes before its right sibling.
Also fix an ordering of the degree $3$ vertices in $T$, which we refer to as {\em branch points} (there are $\ell-2$ of these).

Let $L$ be generated by a context-free grammar $\Gamma$ in Chomsky normal form with nonterminals $\mathcal N$,
terminals $\mathcal X$ and productions $\mathcal P$.
Our indexed grammar $\Gamma_T$ is defined using the following symbols:
\begin{itemize}
\item terminals $\mathcal X$,
\item flags ${\mathcal F} = \{\$,\#_i,A_\alpha,A_\omega,A_L,A_R,a_\omega \mid 1\leq i\leq N, \A\in \mathcal N, a\in \mathcal X\}$,
\item start symbol $\s_0$ and remaining nonterminals \[\{\A_e\ \mid \A\in {\mathcal N}, e\in E(T)\} \cup U \cup \ol{U},\]
where \[U = \left\{\X_s, \X_{s,B}, \X_{s,BC}, \M, \M_{ij} \mid s, i\in \{1,\ldots \ell\}, j\in \{1,\ldots, k_i\}, \B,\C\in \mathcal N\right\}\]
($k_i$ will be defined later) and $\ol{U} = \{\ol{\C} \mid \C\in U\}$.
\end{itemize}
 
We begin by constructing a potential $T$-skeleton by replacing 
each edge $e$ of $T$ by a valid edge-skeleton $P_e$ in $\Gamma$.  
We will check the consistency of the branch points (degree $3$ vertices of $T$) later.
For now the initial symbols of all the edge-skeletons other than the one for the initial edge $e_1$ (which has initial symbol $\s_0$) 
will be chosen non-deterministically.  We do not need to remember the labels of the vertices on the frames of the edge-skeletons $P_e$,
other than the initial and terminal vertices.  We can therefore create and store all the important information about 
our $T$-skeleton using the following rules.  The subscripts $\alpha$ and $\omega$ denote the beginning and end respectively of a
edge-skeleton $P_e$, while the subscripts $L$ and $R$ record whether a vertex lies to the left or right of its frame.
[Note:  Throughout this proof, we will sometimes use brackets around nonterminals to increase legibility.  The brackets have no 
meaning in the grammar.]

\begin{enumerate}
\item $\s_0 \ra (\s_{e_1})^{S_\alpha}$, 
\item $\A_e \ra (\B_e)^{C_R} \mid (\C_e)^{B_L}$ for each edge $e$ and production $\A\ra \B\C$ in $\mathcal P$,
\item $\A_{e_i} \ra (\B_{e_{i+1}})^{B_\alpha \#_i A_\omega}$ for each $\A,\B\in \mathcal N$ and $1\leq i\leq m$,
\item $\A_{e_i} \ra (\B_{e_{i+1}})^{B_\alpha \#_i a_\omega}$ for each production $\A\ra a$ in $\mathcal P$, $\B\in \mathcal N$
and $m+1\leq i\leq N-1$,
\item $\A_{e_N} \ra (\M)^{\#_N A_\omega}$ for each $\A\in \mathcal N$.
\end{enumerate} 

After a sequence of these productions, terminating with (5), the string produced is a single nonterminal $\M$ with flag
\[ \#_N \omega_N v_N \alpha_N \#_{N-1} \ldots \#_{2} \omega_2 v_2 \alpha_2 \#_{1} \omega_1 v_1 \alpha_1 \$. \]
The section of the flag in between $\#_i$ and $\#_{i-1}$ contains the information about the path-skeleton $P_{e_i}$.
The symbols $\alpha_i$ and $\omega_i$ correspond to the initial and final vertices of $P_{e_i}$ respectively,
and we have $\alpha_i = B_\alpha$ for some nonterminal $\B$ (in particular, $\alpha_1 = S_\alpha$), 
while $\omega_i = A_\omega$ for some nonterminal $\A$
if $1\leq i\leq m$, and $\omega_i = a_\omega$ for some terminal $a$ otherwise (that is, if $e_i$ is a leaf).
The $v_i$ are words in $\{A_L, A_R\mid A\in \mathcal N\}^*$ encoding -- in reverse -- the sequence of vertices off the main path in $P_{e_i}$,
with the subscripts $L$ and $R$ denoting a vertex lying to the left or right of the path respectively.

Having produced a flag corresponding to a potential $T$-skeleton, we now need to check that  
this is indeed a valid $T$-skeleton in $\Gamma$.  The only potential problems are at the branch points.  If $e_p$ is an 
edge with left child $e_q$ and right child $e_r$, then we need to check that $\mathcal P$ contains a production
$\A\rightarrow \B\C$, where $\omega_p = A_\omega$, $\alpha_q = B_\alpha$ and $\alpha_r = C_\alpha$.
In order to do this, we create a `check symbol' $\X_i$ for each branch point.

\begin{enumerate}[resume]
\item $\M\ra  \ol{\M} \X_1 \ldots \X_{\ell-2}.$
\end{enumerate}

Let $e_p$ be the edge coming down into the $i$-th branch point and let $e_q$ and $e_r$ be its left and right children respectively.
Recall that we have chosen our ordering on the edges in such a way that $p<q<r$.
Call a nonterminal $\A$ {\em $f$-ready}, for some flag $f$, if $\A^g \ra \A$ for all $g\in {\mathcal F}\setminus \{f\}$.
Now $\X_i$ checks for a valid production at the $i$-th branch point.
Informally, the idea is that $\X_i$ searches the flag for the symbols $\alpha_r, \alpha_q$ and $\omega_q$ (which will occur in that 
order) and stores $C$ and $B$, where $\alpha_r = C_\alpha$, $\alpha_q = B_\alpha$, finally outputting the empty word if and only 
if $\A\ra \B\C$ is in $\mathcal P$, where $\omega_p = A_p$.  Formally, this is achieved via the following productions, requiring quite a few extra nonterminals:

\begin{enumerate}[resume]
\item $\X_i$ is $\#_r$-ready, with $(\X_i)_{\#_r} \ra \ol{\X}_i$,
\item $\ol{\X}_i$ is ready for the next $\alpha$-subscripted flag (which will be $\alpha_r$), 
and $(\ol{\X}_i)^{C_\alpha} \ra \X_{i,C}$,
\item $\X_{i,C}$ is $\#_q$-ready, with $(\X_{i,C})^{\#_q} \ra \ol{\X}_{i,C}$,
\item $\ol{\X}_{i,C}$ is ready for the next $\alpha$-subscripted flag (which will be $\alpha_q$),
and $(\ol{\X}_{i,C})^{B_\alpha} \ra \X_{i,BC}$,
\item $\X_{i,BC}$ is $\#_p$-ready, with $(\X_{i,BC})^{\#_p}\ra \ol{\X}_{i,BC}$,
\item The next flag will be $\omega_p$, so $\ol{\X}_{i,BC}$ can now finally check the validity of the 
branch point, by having productions $(\ol{\X}_{i,BC})^{A_\omega} \ra \e$ for all productions $\A\ra \B\C$ in $\mathcal P$.
\end{enumerate}

These productions ensure that once the symbols $\X_i$ are introduced, the derivation will never terminate unless the 
potential $T$-skeleton in the flag has valid connections at all branch points, and hence is a valid $T$-skeleton,
in which case all the symbols $\X_i$ produce the empty word.

Finally, we `unpack' the $T$-skeleton from the flag, to produce words $\tau(w)$, where $w$ is a word in $L$ arising from the $T$-skeleton.
Let $e_i^L$ and $e_i^R$ be the left and right hand sides of the edge $e_i$ respectively. 
Consider the `outline' of $T$, which is a directed path $o(T)$ drawn around the outside of $T$, beginning on the left hand side of the root and ending 
on the right hand side of the root, divided into labelled segments corresponding to the edge-sides.

\begin{figure}[h!]\tiny
 \begin{tikzpicture}[scale=.5]
\draw[decorate] (0,0) -- (6,6) -- (6,9);
\draw[decorate] (3,3) -- (6,0);
\draw[decorate] (6,6) -- (9,3);

\draw (5,8) node {$e_1^L$};
\draw (7,8) node {$e_1^R$};
\draw (3.5,5.5) node {$e_2^L$};
\draw (5.5,3.5) node {$e_2^R$};
\draw (.5,2.5) node {$e_3^L$};
\draw (2.5,.5) node {$e_3^R$};
\draw (3.5,.5) node {$e_4^L$};
\draw (6.5,1.5) node {$e_4^R$};
\draw (6.5, 3.5) node {$e_5^L$};
\draw (8.5, 5.5) node {$e_5^R$};

\draw[blue] (5.5,6.5) node {$\bullet$};
\draw[blue] (6.5,6.5) node {$\bullet$};
\draw[blue] (6,5) node {$\bullet$};
\draw[blue] (2.5,3.5) node {$\bullet$};
\draw[blue] (4,3) node {$\bullet$};
\draw[blue] (3,2) node {$\bullet$};
\draw[blue] (-.5,-.5) node {$\bullet$};
\draw[blue] (6.5,-.5) node {$\bullet$};
\draw[blue] (9.5,2.5) node {$\bullet$};

\draw [red, directed] plot [smooth] coordinates {(5.5,9)  (5.5,6.5)};
\draw [red] plot [smooth] coordinates { (5.5,6.5)  (2.5,3.5)};
\draw [red] plot [smooth] coordinates { (2.5,3.5)  (0,1)  (-.5,-.5)  };
\draw [red] plot [smooth] coordinates {  (-.5,-.5)  (1,0) (3,2)};
\draw [red] plot [smooth] coordinates { (3,2) (5,0) (6.5,-.5)  };
\draw [red] plot [smooth] coordinates {(6.5,-.5) (6,1) (4,3)};
\draw [red] plot [smooth] coordinates {(4,3) (6,5) };
\draw [red] plot [smooth] coordinates {(6,5) (8,3)  (9.5,2.5)};
\draw [red] plot [smooth] coordinates { (9.5,2.5)  (9,4) (6.5,6.5)};
\draw [red] plot [smooth] coordinates {  (6.5,6.5)  (6.5,9)};

\end{tikzpicture}
\caption{Outline.
\label{fig:outline}}
\end{figure}

The segment of $o(T)$ lying between the $(i-1)$-th and $i$-th leaves of $T$ corresponds to $w_i$
(here we regard the root as both the $0$-th and $\ell$-th leaf).
For $1\leq i\leq \ell$, let $k_i$ be the length of this segment.
We write $w_i\sim f_1\ldots f_{k_i}$ if the labels on the segment of $o(T)$ corresponding to $w_i$, in order, are $f_1,\ldots,f_{k_i}$.
For $1\leq i\leq \ell$, define $\rho_i:\{1,\ldots,k_i\}\ra \{1,\ldots,N\}$ and $d_{ij}\in \{L,R\}$ such that 
if $w_i\sim f_1\ldots f_{k_i}$, then $f_j = e_{\rho_i(j)}^{d_{ij}}$.
Then if $w$ is a word in $L$ produced from our $T$-skeleton, we have $w = w_1\ldots w_\ell$, where $w_i$ arises from the edge-sides 
$e_{\rho_i(1)}^{d_{i1}}, \ldots, e_{\rho_i(k_i)}^{d_{ik_i}}$ in that order.  Note that $d_{i1}=R$ for all $i\geq 2$.
We prepare to unpack the flag in the correct order to produce $\tau(w)$ as follows:

\begin{enumerate}[resume]
\item $\ol{\M} \ra \ol{\M}_{\tau(1)1} \ldots \ol{\M}_{\tau(1)k_1} \ldots \ol{\M}_{\tau(\ell) 1} \ldots \ol{\M}_{\tau(\ell) k_\ell}$.
\end{enumerate}

And finally, the actual unpacking occurs, using the following productions.

\begin{enumerate}[resume]
\item $\ol{\M}_{ij}$ is $\#_{\rho_i(j)}$-ready, with $(\ol{\M}_{ij})^{\#_{\rho_i(j)}} \ra \M_{ij}$,
\item $(\M_{ij})^{a_\omega} \ra a \M_{ij}$ if $d_{ij} = R$ and $(\M_{ij})^{a_\omega}\ra \e$ if $d_{ij}=L$,
\item $(\M_{ij})^{A_L} \ra \M_{ij} \tilde{\A}$ if $d_{ij} = L$, and $(\M_{ij})^{A_L} \ra \M_{ij}$ if $d_{ij} = R$,
\item $(\M_{ij})^{A_R} \ra \tilde{\A} \M_{ij}$ if $d_{ij} = R$, and $(\M_{ij})^{A_R} \ra \M_{ij}$ if $d_{ij} = L$,
\item $(\M_{ij})^{B_\alpha} \ra \e$ and $(\M_{ij})^{A_\omega} \ra \e$,
\item for all $\A\in \mathcal N$, $\tilde{\A}$ is $\$$-ready and $\tilde{\A}^\$ \ra \A$,
\item all productions in $\mathcal P$.
\end{enumerate}

The productions allow $\ol{\M}_{ij}$ to find the section of the flag corresponding to the edge $e_{\rho_i(j)}$, 
from which the $j$-th segment of $w_i$ arises, and then unpack the relevant parts (i.e. the vertices from the left or right side) 
of that section in the correct direction.  The flag contains the information about each edge-skeleton $P_e$ in reverse order.
For subwords generated by $e^L$, this is the `wrong' order and so we need to unpack the left-hand vertices of $P_e$ 
to the right, while for subwords generated by $e^R$, this is the correct order and so we unpack the right-hand 
vertices of $P_e$ to the left.
A flag $a_\omega$ belongs to the right side of its edge, by our convention for the partition defined by the
tree-shape $T$.  When we reach a symbol $B_\alpha$, we have finished recovering the relevant segment, and so we output $\e$.
Finally, we produce an appropriate subword of a word in $L$ using productions from $\mathcal P$.
Once we have produced a word consisting entirely of terminals, we have $w_{\tau(1)}\ldots w_{\tau(\ell)}$ 
for some partitioned word $w_1 | \ldots | w_\ell$ in $L(B)$, where $B$ is the $T$-skeleton encoded in the flag.
All such words can be produced in this way, and so the language generated by $\Gamma_T$ is indeed the union of all $L_\tau(B)$
with $B$ a $T$-skeleton in $\Gamma$.  Hence $L_\tau$ is the union of finitely many indexed languages and is thus itself indexed.
\end{proof}

Our main result follows immediately. 

\begin{cor}
Let $\sigma\in S_k$ be any permutation.  If $L$ is context-free, then $\sigma(L)$ is indexed.
\end{cor}

\begin{cor}\label{cor:secondmain}
Let $k$ be a positive integer.
If $L$ is context-free, then $C^k(L)$ is indexed (and context-free if $k=1,2$).
\end{cor}

\bibliography{refs} \bibliographystyle{plain}
\end{document}